\documentclass[12pt]{article}

\usepackage{url}
\usepackage{bbm}
\usepackage{mathtools}
\usepackage{amssymb}
\usepackage{amsthm}
\usepackage{empheq}
\usepackage{latexsym}
\usepackage{eurosym}
\usepackage{dsfont}
\usepackage{appendix}
\usepackage{color} 
\usepackage[unicode]{hyperref}
\usepackage{frcursive}
\usepackage[utf8]{inputenc}
\usepackage[T1]{fontenc}
\usepackage{geometry}
\usepackage{multirow}
\usepackage{todonotes}
\usepackage{lmodern}
\usepackage{anyfontsize}
\usepackage{pgfplots}
\usepackage{stmaryrd}
\usepackage{float}

\usepackage{graphicx}
\usepackage{caption}
\usepackage{subcaption}
\usepackage{enumerate}
\usepackage{breakcites}
\graphicspath{ {images/} }

\pgfplotsset{compat=newest}

\definecolor{red}{rgb}{0.7,0.15,0.15}
\definecolor{green}{rgb}{0,0.5,0}
\definecolor{blue}{rgb}{0,0,0.7}
\hypersetup{colorlinks, linkcolor={red},citecolor={green}, urlcolor={blue}}
            
\makeatletter \@addtoreset{equation}{section}

\newtheorem{theorem}{Theorem}
\newtheorem{theorem2}{Theorem}[section]

\newtheorem{lemma}[theorem2]{Lemma}
\newtheorem{proposition}[theorem2]{Proposition}

\newtheorem{definition}[theorem2]{Definition}
\newtheorem{remark}[theorem2]{Remark}
\newcommand{\comment}[1]{}
\DeclareUnicodeCharacter{014D}{\=o}
\def \E{\mathbb{E}}

\def \N{\mathbb{N}}

\def \P{\mathbb{P}}

\def \R{\mathbb{R}}

\def \V{\mathbb{V}}

\def\Cc{{\cal C}}

\def\Nc{{\cal N}}

\title{Equilibria and incentives for illiquid auction markets}
\author{Joffrey {\sc Derchu}\footnote{\'Ecole Polytechnique, CMAP, 91128, Palaiseau, France, joffrey.derchu@polytechnique.edu},~ Dimitrios {\sc Kavvathas}\footnote{Harmony Advisors Ltd.; Dimitrios.kavvathas@harmonyadvisors.com},~ Thibaut {\sc Mastrolia}\footnote{Department of Industrial Engineering and Operations Research, UC Berkeley; mastrolia@berkeley.edu}~~\\and Mathieu {\sc Rosenbaum}\footnote{\'Ecole Polytechnique, CMAP; mathieu.rosenbaum@polytechnique.edu} \footnote{This work benefits from the financial support of the Chaires Analytics and Models for Regulation, Deep finance and statistics, Machine learning and systematic methods.}}

\begin{document}

\maketitle
\begin{abstract}
We study a toy two-player game for periodic double auction markets to generate liquidity. The game has imperfect information, which allows us to link market spreads with signal strength. We characterize Nash equilibria in cases with or without incentives from the exchange. This enables us to derive new insights about price formation and incentives design. We show in particular that without any incentives, the market is inefficient and does not lead to any trade between market participants. We however prove that quadratic fees indexed on each players half spread leads to a transaction and we propose a quantitative value for the optimal fees that the exchange has to propose in this model to generate liquidity.   
\end{abstract}

\section{Introduction}

Auctions are becoming an increasingly popular trading mechanism. In such market design, market participants typically send
market and limit orders during some (randomized) period of time. At the end of this period, a clearing price is set so that
it maximises the traded volume. Some exchanges offer auctions in a periodic way throughout the day, such as BATS-CBOE for
European equities. More generally,  many platforms use auctions at the beginning and at the end of the day, often collecting a very
significant part of the daily traded volume at these specific times. The large interest in auctions in the recent years is notably
due to the role of high frequency trading in the modification of the financial markets ecosystem.
In some influential papers such as \cite{budish2015high,du2017optimal,farmer2012review,madhavan1992trading,wah2013latency} the authors explore flaws of limit order books and argue that high
frequency traders operating in a continuous limit order book generate negative externalities for investors
due to their speed advantage. By their very definition, auctions mechanically slow down the market and therefore
are a remedy to some of the supposed inherent flaws of continuous limit order books associated to ultra fast trading. Inspired by \cite{du2017optimal,Duffie2015SizeD,Antill2017AugmentingMW}, \cite{jusselin2019optimal} shows that from a price discovery viewpoint, the optimal auction duration is between 1 and 10 minutes for typical stocks. In addition they demonstrate that continuous limit order books are in general a sub-optimal market mechanism. Combination of auctions and continuous limit order books are investigated in \cite{derchu2020ahead}. Furthermore, for illiquid assets, offering continuous trading is probably not required to obtain a satisfactory price formation process. Therefore, auctions appears as a natural and cost efficient design for an exchange.\\

In this work, we are interested in understanding the behavior of market participants during an auction in a illiquid market. We show in particular how incentives allow the exchange to create new market equilibria and generate liquidity. By illiquid we mean a market where trading frequency is low, volumes are small, market data is limited and information is scarce.
In this framework, we use a notion of efficient price, to be understood as a partially known long term value of the asset, also often called the ``fair'' price. Since we consider an illiquid market, at a given time point, it may be significantly different from the last traded price.\\

We model this market by a very simplified one period framework with two participants, one buyer and one seller, taking part to an auction organized by an exchange. At the time of the auction, we assume that the increment between the fair price, denoted by $P_{\infty}$ and the last traded price is given by a centered Gaussian variable with variance $v^2$. The parameter $v$ can be interpreted as the degree of illiquidity of our market. A large $v$ can for instance translate the fact that auctions are not very frequent, leading to large fluctuations of the fair price between them. It can also mean that the asset is simply very volatile, for example because of the uncertainty about its value due to limited available information.\\

We assume both participants have access to a proxy of the fair price at the beginning of the auction, denoted by $P_{\infty|a}$ for the seller and $P_{\infty|b}$ for the buyer. The standard deviations of the error of this proxy are drawn from independent random variables, denoted by $\sigma_a$ and $\sigma_b$. These standard deviations are specific to each market participant and represent their degree of sophistication. Using his personal information, the buyer agent computes a maximal price at which he is willing to buy and send a limit order, and symmetrically for the seller one. Inspired by standard quoting mechanisms, we suppose that the agents build their quote as a deviation from their proxy depending on their uncertainty about the quality of their proxy. Therefore these quotes take the form  $P_{\infty|i}\pm \delta_i(\sigma_i)$ for $i=a,b$. More precisely, we are interested in the symmetric case where market participants have the same pricing rule and consequently the function $\delta_i$ does actually not depend on $i$. If the buy limit order price is larger than the sell limit order price, a trade happens at the midprice. So, for a given function $\delta$, each (risk neutral) agent optimizes simultaneously the price of its limit order by maximizing or minimizing the expected value of the difference between the obtained midprice and the fair price, taking into account the probability that no trade happens.\\

The goal of this paper is to study the situations where such game leads to a viable market. To do so, we aim at finding Nash equilibria. A Nash equilibrium in our context simply consists in the existence of a function $\delta$ such that both market participants can optimize at the same time and have no interest to deviate from the computed optima. We actually show that in such framework, Nash equilibria usually do not exist. This is bad news from an exchange or regulator perspective as it tells us that even in a very simplified setting, we can expect unstable behaviours from market participants. To remedy this, we consider that the exchange is able to put in place a make-take fees system, where he rewards or penalizes market participants with respect to the quality of their quotes. We first study the case where the penalty on market participants depends on the price formation process: participants are penalized if the auction clearing price is far from the efficient price. Then we consider the situation where the penalty is computed from the displayed spreads. Interestingly, we establish that one can recover Nash equilibria with the second mechanism while the first one does not enable us to do so. We also explain how to  design optimally the spread penalty leading to Nash equilibrium.\\

The paper is organized as follows. We start by introducing the model and the assumptions in Section \ref{sec::modeleq}. Then we consider the case where an exchange imposes penalties on the players in Section \ref{sec::pen}. Section \ref{sec::num} gathers the numerical results. We study some extensions of our framework in Section \ref{sec::ext}. Finally, some proofs are relegated to the appendix.

\section{The game}\label{sec::modeleq}
We consider two players, Player $a$ and Player $b$ aiming at trade an asset. Player $a$ sends a sell limit order at a price $P^a\in\R$ while Player $b$ sends a buy limit order at the price $P^b\in\R$. Both players announce their price simultaneously. If $P^a\leq P^b$, a trade occurs at the price $\frac{P^a+P^b}{2}$. In this case, player $a$ sells one unit at the price $\frac{P^a+P^b}{2}$ to player $b$. From now on, Player $a$ is always selling ($a$ stands for `ask') and Player $b$ is always buying ($b$ stands for `bid'). We could easily extend our framework to the case where our two players send both buy and sell orders.\\

 We assume that there exists on this market an  efficient price of the asset denoted by $P_\infty$. At the time of the auction, the increment between this fair price and the last traded price of the previous auction is normally distributed with mean 0 and variance $v^2$ with $v\in\R$ and possibly $v\to +\infty$. The parameter $v$ represents the degree of illiquidity of the market. A large $v$ indicates that auctions are not very frequent leading to possibly large fluctuation between the efficient price and the auction price or that the asset is very volatile.

\subsection{Efficient price estimation and market statistics}

Due to the illiquid nature of our market, our players do not have access to $P_\infty$. They can only estimate it based on some private information. Each player $i=a,b$ compute his own private estimate $P_{\infty|i}$ of $P_\infty$ defined by
\begin{align*}
    P_{\infty|a}=P_\infty+\boldsymbol\sigma_a\epsilon_a\\
    P_{\infty|b}=P_\infty+\boldsymbol\sigma_b\epsilon_b
\end{align*}
where\footnote{From now on, bold letters are used to emphasize certain random variables.}
\begin{itemize}
    \item $\boldsymbol\sigma_a, \boldsymbol\sigma_b$ denote player $i$'s uncertainty in his estimate. We assume that $\boldsymbol\sigma_a$ and $\boldsymbol\sigma_b$ are uniformly distributed on  $[\sigma_-,\sigma_+]$, for some bounds $0<\sigma_-<\sigma_+$.  \item $\epsilon_a$ and $\epsilon_b$ are normally distributed with zero mean and unit variance such that $\text{Corr}(\epsilon_a,\epsilon_b)=\rho\in(-1,1)$.
    \item The variables $P_\infty,\boldsymbol\sigma_a, \boldsymbol\sigma_b$ and the couple $(\epsilon_a,\epsilon_b)$ are all independent from one another. 
\end{itemize}

We assume that the only information available to each player $i$ is $P_{\infty|i}$ and $\boldsymbol\sigma_i$. In particular the players do not know the efficient price $P_\infty$, the other player's signal $P_{\infty|i}$ nor the other player's confidence in his signal $\boldsymbol\sigma_i$. Recalling that bold symbol stands for random variables and to simplify notations, we set for some $\sigma_a,\sigma_b\in[\sigma_-,\sigma_+
] $
\[
\Sigma_\rho:= \sqrt{\sigma_a^2+\sigma_b^2-2\rho\sigma_a\sigma_b}, \quad  \boldsymbol\Sigma_\rho:=\sqrt{ \boldsymbol\sigma_a^2+\boldsymbol\sigma_b^2-2\rho\boldsymbol\sigma_a\boldsymbol\sigma_b},\quad \boldsymbol\Sigma^b_\rho:=\sqrt{ \sigma_a^2+\boldsymbol\sigma_b^2-2\rho\sigma_a\boldsymbol\sigma_b},
\]

\[Q_\rho:=  \frac{\frac{1}{\sigma_b^2}-\frac{\rho}{\sigma_a\sigma_b}}{\frac{1}{\sigma_a^2}+\frac{1}{\sigma_b^2}-2\frac{\rho}{\sigma_a\sigma_b}}=\frac{\sigma_a-\rho\sigma_b}{\Sigma_\rho^2}, \quad \boldsymbol Q^b_\rho:= \frac{\frac{1}{\boldsymbol\sigma_b^2}-\frac{\rho}{\sigma_a\boldsymbol\sigma_b}}{\frac{1}{\sigma_a^2}+\frac{1}{\boldsymbol\sigma_b^2}-2\frac{\rho}{\sigma_a\boldsymbol\sigma_b}}=\frac{\sigma_a-\rho\boldsymbol\sigma_b}{|\boldsymbol\Sigma^b_\rho|^2}, 
\]
and
\[ \tilde Q_\rho:=\frac{\frac{1}{\sigma_a^2}-\frac{\rho}{\sigma_a\sigma_b}}{\frac{1}{\sigma_a^2}+\frac{1}{\sigma_b^2}-2\frac{\rho}{\sigma_a\sigma_b}}=\frac{\sigma_b-\rho\sigma_a}{\Sigma_\rho^2}.\]

As explained in the introduction, we assume that the way market participants form their prices $P^a,P^b$ is as follows
\[ P^a=P_{\infty|a}+\delta(\boldsymbol\sigma_a) \text{ and } P^b=P_{\infty|b}-\delta(\boldsymbol\sigma_b),\text{ for some function }\delta: [\sigma_-,\sigma_+]\longrightarrow \mathbb R.\]

Given a strategy $\delta$, statistics on the market quality can be easily derived as stated in the proposition below. 
\begin{proposition}\label{prop::stats}
Let $\delta$ be fixed, then we can compute the following quantities
\begin{itemize}
  \item Spread:   $\mathbb E[P^a-P^b]=2\E[\delta(\boldsymbol\sigma_a)]$
  \item Variance of the spread:   $\V[P^a-P^b]=2(\E[\boldsymbol\sigma_a^2]-\rho\E[\boldsymbol\sigma_a]^2+\V[\delta(\boldsymbol\sigma_a)])$
 \item Average mid price-efficient price error:   $ \E[\frac{P^a+P^b}{2}-P_\infty] = 0$
   \item Variance of the mid price-efficient price error:   $\V[\frac{P^a+P^b}{2}-P_\infty]= \frac{\E[\boldsymbol\sigma_a^2]+\rho\E[\boldsymbol\sigma_a]^2+\V[\delta(\boldsymbol\sigma_a)]}{2}$ 
  \item Trade occurrence: $ \P[P^a\leq P^b] = 1-\E[\Phi(\frac{\delta(\boldsymbol\sigma_a)+\delta(\boldsymbol\sigma_b)}{\boldsymbol\Sigma_\rho})].$ 
  \end{itemize}
\end{proposition}
We note that those results hold when $v$ is finite and at the limit $v\to+\infty$.\\

The average error on the efficient price is zero which shows that if there is an equilibrium a reasonable price formation process takes place. The variance of the spread and the variance of the error on the efficient price are equal up to a multiplicative constant. Interestingly, only the variance of $\delta(\boldsymbol\sigma_a)$ appears in the variance of the error on the efficient price. In other words the relevance of the mid-price as an estimator of the efficient price is not impaired if the spreads are uniformly increased (if we add some constant to $\delta$). The quality of the mid-price proxy decreases if the  spreads are transformed linearly (if we multiply $\delta$ by some constant larger than $1$). Note also that both variances are lower bounded (by $2\V[\boldsymbol\sigma_a^2]$ and $\frac{\V[\boldsymbol\sigma_a^2]}{2}$ respectively): this implies that the market quality is limited by the quality of its participants, and no amount of incentives to change $\delta$ could make the market better than this limit. The number of trades is related in a non-trivial way to $\delta$, though we still have the obvious result that a smaller $\delta$ leads to more trades. Positive correlation tends to fix the spread but damages the price formation process.

\subsection{Optimization and equilibrium}

We assume that both players send their orders simultaneously without any communication between them.\\

Given $\boldsymbol\sigma_a$ and $P_{\infty|a}$ and that Player $b$ sets $P^b=P_{\infty|b}-\delta(\boldsymbol\sigma_b$), Player $a$ wants to maximize over $x=P^a-P_{\infty|a}$ his expected payoff
\begin{align*}
    \E[(\frac{P_{\infty|a}+x+P_{\infty|b}-\delta(\boldsymbol\sigma_b)}{2}-P_\infty)\mathbf{1}_{P_{\infty|a}+x\leq P_{\infty|b}-\delta(\boldsymbol\sigma_b)}| (P_{\infty|a},\boldsymbol\sigma_a)].
\end{align*}

We have the following proposition whose proof is in Appendix \ref{app::prop::lim}. 
\begin{proposition}\label{prop::lim}
Let $\delta$, $\sigma_a\in[\sigma_-,\sigma_+]$ and $p_{\infty|a}\in\R$ be fixed. Then 
\begin{align*}
    x\in\R\mapsto\E[(\frac{P_{\infty|a}+x+P_{\infty|b}-\delta(\boldsymbol\sigma_b)}{2}-P_\infty)\mathbf{1}_{P_{\infty|a}+x\leq P_{\infty|b}-\delta(\boldsymbol\sigma_b)}| (P_{\infty|a}=p_{\infty|a},\boldsymbol\sigma_a=\sigma_a)]
\end{align*}
converges uniformly to 
\begin{align*}
    e^{\sigma_a}_\delta: x\mapsto \frac{1}{\sigma_+-\sigma_-}\int_{\sigma_-}^{\sigma^+}\big(&\frac{x-\delta(\sigma_b)}{2}(1-\Phi(\frac{x+\delta(\sigma_b)}{\Sigma_\rho}))+\Sigma_\rho(\frac{1}{2}-Q_\rho)\Phi'(\frac{x+\delta(\sigma_b)}{\Sigma_\rho})\big)d\sigma_b
\end{align*}
as $v\to+\infty$, with $\Phi$ the cumulative distribution function of a standard Gaussian variable. Note that $e^{\sigma_a}_\delta$ does not depend on $p_{\infty|a}$.
\end{proposition}

Note that in the case where $v$ is finite, $ e^{\sigma_a}_\delta$ could depend on $P_{\infty|a}$. However, if we take $v\to +\infty$, the expected payoff converges to a limit which does not depend on $P_{\infty|a}$ as this limit corresponds to an extremely illiquid market.

\begin{remark}\label{rem:finitedelta}
We can rewrite \begin{align*}
    e^{\sigma_a}_\delta(x)&= \frac{1}{\sigma_+-\sigma_-}\int_{\sigma_-}^{\sigma^+}\big(\frac{x-\delta(\sigma_b)}{2}(1-\Phi(\frac{x+\delta(\sigma_b)}{\Sigma_\rho}))+\Sigma_\rho(\frac{1}{2}-Q_\rho)\Phi'(\frac{x+\delta(\sigma_b)}{\Sigma_\rho})\big)d\sigma_b\\
    &=\E[\big(\frac{x-\delta(\boldsymbol\sigma_b)}{2}+\tilde\epsilon_\infty\frac{1}{\sqrt{\frac{1}{\sigma_a^2}+\frac{1}{\boldsymbol\sigma_b^2}}}+\tilde\epsilon_{ab}\boldsymbol\Sigma^b_\rho(\frac{1}{2}-\boldsymbol Q_\rho^b)\big)\mathbf{1}_{\tilde\epsilon_{ab}\geq \frac{x+\delta(\boldsymbol\sigma_b)}{\boldsymbol \Sigma_\rho^b}}],\\
\end{align*}
where $ ( \tilde\epsilon_{ab},
    \tilde \epsilon_{\infty})$ is a bivariate normal random variable centered with unit covariance matrix.
\end{remark}

We now define a Nash equilibrium in our setting.
\begin{definition}[Nash equilibrium (NE)]\label{nash:definiiton}
Let $\delta$ be a function from $[\sigma_-,\sigma_+]$ to $\R_+^*$. We say that $\delta$ is a Nash equilibrium (NE) in the case $v\to \infty$ if, for all $\sigma_a\in[\sigma_-,\sigma_+]$, we have\footnote{By symmetry we do not repeat this condition for Player $b$. If $\delta$ is an optimizer for Player $a$, it is also optimal for Player $b$.}
\begin{align*}
    \delta(\sigma_a) \in \arg\sup e^{\sigma_a}_\delta.
\end{align*}
\end{definition}

From Remark \ref{rem:finitedelta} we deduce the following important lemma which is a first step toward finding a Nash equilibrium.

\begin{lemma}
Let $\delta$ be some non-negative function on $[\sigma_-, \sigma_+]$, $\sigma_a\in[\sigma_-, \sigma_+]$. Assume that $\rho\leq\frac{\sigma_-}{\sigma_+}$. Then $e^{\sigma_a}_\delta$ attains its maximum on $\R_+$.
\end{lemma}

\subsection{No incentive, no equilibrium, no liquidity}
The main result of this section is that, in our framework with symmetric players and without any incentives, there is no Nash equilibrium.
\begin{proposition}\label{prop::noNE}
There is no Nash equilibrium  as defined in Definition \ref{nash:definiiton}. 
\end{proposition}
\begin{proof}
If $\sigma_+=\sigma_-$, we need $1-\Phi(\frac{\sqrt{2}\delta(\sigma_+)}{\sigma_+\sqrt{1-\rho^2}})=0$ which is impossible with $\delta(\sigma_+)$ finite.\\

If $\sigma_+>\sigma_-$, we must have
\begin{align*}
    &(e^{\sigma_a}_\delta)'(\delta(\sigma_a))\\
    &= \frac{1}{\sigma_+-\sigma_-}\int_{\sigma_-}^{\sigma^+}\big(\frac{1}{2}(1-\Phi(\frac{\delta(\sigma_a)+\delta(\sigma_b)}{\Sigma_\rho}))+\frac{1}{\Sigma_\rho}(-\delta(\sigma_a)\tilde Q_\rho+\delta(\sigma_b)Q_\rho)\Phi'(\frac{\delta(\sigma_a)+\delta(\sigma_b)}{\Sigma_\rho})\big)d\sigma_b\\
    &=0
\end{align*}
for all $\sigma_a\in[\sigma_-,\sigma_+]$. Thus
\begin{align*}
    &\int_{\sigma_-}^{\sigma^+}\int_{\sigma_-}^{\sigma^+}\big(\frac{1}{2}(1-\Phi(\frac{\delta(\sigma_a)+\delta(\sigma_b)}{\Sigma_\rho}))\big)d\sigma_a d\sigma_b=0
\end{align*}
which is impossible with $\delta$ finite.
\end{proof}

To mitigate the fact that there is no equilibrium in our framework, we introduce some frictions under the form of transaction fees imposed by the exchange, hoping that it can lead to a more viable market.

\section{Optimal transaction costs to create liquidity }\label{sec::pen}

We introduce penalties in the players payoffs. This translates real-world liquidity provision programs.

\subsection{Price discovery penalty: inefficient policy}\label{subsec::quadfeesmidprice}
In this section, we assume that the exchange is able to know $P_\infty$. For example, one can thing of $P_\infty$ as the closing price pf the day.
We suppose that both players pay a penalty $\gamma(\frac{P^a+P^b}{2}-P_\infty)^2$ with $\gamma>0$. The problem of Player a is to maximize on $x=P^a-P_{\infty|a}$ the penalized payoff 

\begin{align*}
   \E[(\frac{P_{\infty|a}+x+P_{\infty|b}-\delta(\boldsymbol\sigma_b)}{2}-P_\infty)\mathbf{1}_{P_{\infty|a}+x\leq P_{\infty|b}-\delta(\boldsymbol\sigma_b)}-\gamma(\frac{P^a+P^b}{2}-P_\infty)^2| (P_{\infty|a},\boldsymbol\sigma_a)].
\end{align*}

Then we have the following result.
\begin{proposition}\label{prop::penquaderr}
Let $\gamma>0$. Then there is no NE with finite $\delta$.
\end{proposition}
\begin{proof}
Using the same notations as previously we have
\begin{align*}
    \E[(\frac{P^a+P^b}{2}-P_\infty)^2| (P_{\infty|a},\sigma_a)] = \frac{1}{4}(x^2-2x\E[\delta(\sigma_b)]+\E[\delta(\boldsymbol\sigma_b)^2]+\sigma_a^2+\E[\boldsymbol\sigma_b^2]+2\rho\sigma_a\E[\boldsymbol\sigma_b])
\end{align*}
so
\begin{align*}
    e^{\sigma_a}_\delta(x)=& \frac{1}{\sigma_+-\sigma_-}\int_{\sigma_-}^{\sigma^+}\big(\frac{x-\delta(\sigma_b)}{2}(1-\Phi(\frac{x+\delta(\sigma_b)}{\Sigma_\rho}))+\Sigma_\rho(\frac{1}{2}-Q_\rho)\Phi'(\frac{x+\delta(\sigma_b)}{\Sigma_\rho})\big)d\sigma_b\\
    &- \gamma\frac{1}{\sigma_+-\sigma_-}\int_{\sigma_-}^{\sigma^+}(\delta(\sigma_b)^2+\sigma_b^2+2\rho\sigma_a\sigma_b)d\sigma_b-\gamma\sigma_a^2 -\gamma (x^2-2x\frac{1}{\sigma_+-\sigma_-}\int_{\sigma_-}^{\sigma^+}\delta(\sigma_b)d\sigma_b)
\end{align*}
and 
\begin{align*}
    (e^{\sigma_a}_\delta)'(x)=& \frac{1}{\sigma_+-\sigma_-}\int_{\sigma_-}^{\sigma^+}\big(\frac{1}{2}(1-\Phi(\frac{x+\delta(\sigma_b)}{\Sigma_\rho}))+\frac{1}{\Sigma_\rho}(-x\tilde Q_\rho+\delta(\sigma_b)Q_\rho)\Phi'(\frac{x+\delta(\sigma_b)}{\Sigma_\rho})\big)d\sigma_b\\
    &-2\gamma(x-\frac{1}{\sigma_+-\sigma_-}\int_{\sigma_-}^{\sigma^+}\delta(\sigma_b)d\sigma_b),
\end{align*}
recalling that $ \tilde Q_\rho:=\frac{\frac{1}{\sigma_a^2}-\frac{\rho}{\sigma_a\sigma_b}}{\frac{1}{\sigma_a^2}+\frac{1}{\sigma_b^2}-2\frac{\rho}{\sigma_a\sigma_b}}.$
If $\sigma_-=\sigma_+$, $\delta$ describes a NE if and only if
\begin{align*}
    \frac{1}{2}(1-\Phi(\frac{\sqrt{2}\delta(\sigma_+)}{\sigma_+}))=0,
\end{align*}
which cannot be obtained.\\

If $\sigma_-<\sigma_+$, we use the same method as with $\gamma=0$.
\end{proof}
As a consequence, a transaction fees based on the distance between the market mid-price and a player's mid-price does not provide an equilibrium.

\subsection{Quadratic fees in the players' half-spread: liquidity generation}
\label{section:nash}
Assume now that $P_{\infty|i}$ is known by the exchange. This can be done for example if the players quote both bid and ask prices (i.e. the game is played twice).\\

Suppose both players pay a penalty $\gamma(P^i-P_{\infty|i})^2$ with $\gamma>0$. The problem of Player $a$ is to maximize on $x=P^a-P_{\infty|a}$ the penalized payoff 
\begin{align*}
\E[(\frac{P_{\infty|a}+x+P_{\infty|b}-\delta(\boldsymbol\sigma_b)}{2}-P_\infty)\mathbf{1}_{P_{\infty|a}+x\leq P_{\infty|b}-\delta(\boldsymbol\sigma_b)}-\gamma(P^a-P_{\infty|a})^2| (P_{\infty|a},\boldsymbol\sigma_a)].
\end{align*}

 Then we have the following result proved in Appendix \ref{app::th::quadsp}.
\begin{theorem}\label{th::quadsp}
Let $\gamma>0$. Assume that $\rho\leq\frac{\sigma_-}{\sigma_+}$. Then
\begin{itemize}
    \item if $\sigma_-=\sigma_+$ there exists a unique NE with finite $\delta$. In that case $\delta(\sigma_+)$ is given by the unique positive root of the function $x\mapsto \frac{1}{2}(1-\Phi(\frac{\sqrt{2}x}{\sqrt{1-\rho}\sigma_+}))-2\gamma x$ 
    \item if $\sigma_-<\sigma_+$, let $ C=2\frac{1}{2\frac{1-\rho\frac{\sigma_+}{\sigma_-}}{1-(\frac{\sigma_-}{\sigma_+})^2}+1}\frac{1-\rho\frac{\sigma_+}{\sigma_-}}{1-(\frac{\sigma_-}{\sigma_+})^2}\sqrt{\frac{1}{2}\frac{(1-\rho^2)^2\sigma_-^4}{\sigma_+^2}}$: if $\gamma\geq \frac{1}{4C}(1-\Phi(\frac{C}{\sqrt{2}\sigma_+}))+\frac{1}{2\sqrt{2}\sigma_-}\Phi'(\frac{C}{\sqrt{2}\sigma_+})$ there exists a NE with $\delta$ bounded by $C$.\footnote{This bound is not tight and we refer to the proof of this result for more details.}
\end{itemize}
\end{theorem}

\begin{remark}
From the proof of Proposition \ref{prop::penquaderr} we can also note that combining the two previously mentioned penalties leads to a NE. Intuitively, and looking at the derivative of $e_\delta^{\sigma_i}$, adding a quadratic penalty in the mid-price error (as in Subsection \ref{subsec::quadfeesmidprice}) to a quadratic penalty in the spread (as in the current subsection) would tend to reduce the variance of the spread, which could be beneficial, as we have seen in Proposition \ref{prop::stats}.
\end{remark}

\subsection{Optimal transaction fees: the exchange's perspective}
We now consider that the exchange receives quadratic fees in the players' half-spread as presented in Section \ref{section:nash}. Recall that this case generates trade between players. For simplicity, we only consider the case $\sigma_-=\sigma_+$ in this section. The optimal PnL of the exchange becomes
\begin{align}\label{eq::exch_prob}
   \underset{\gamma>0}{\sup}\, \gamma \mathbb E[(P^a-P_{\infty|a})^2+ (P^b-P_{\infty|b})^2]= 2\underset{\gamma>0}{\sup}\, \gamma \delta(\sigma_+)^2
\end{align}
where $\delta$ is the strategy chosen by the players, with which they reach a NE in their game with quadratic penalties with parameter $\gamma$. We can compute the optimal $\gamma$ based on the following theorem.

\begin{theorem}\label{th::excheasy}
Assume $\sigma_-=\sigma_+$. Then \eqref{eq::exch_prob} has a unique solution $\gamma^*$. It is given by $\gamma^*=\frac{\Phi'(y^*)}{2\sqrt{2}\sigma_+}$ where $y^*$ is the unique solution to $1-\Phi(y^*)-y^*\Phi'(y^*)=0$. In this case we also have $\delta(\sigma_+)=\frac{\sigma y^*}{\sqrt{2}}$.
\end{theorem}
\begin{proof}
The optimization problem is
\begin{align*}
    \underset{\gamma>0}{\sup}\, \gamma \delta^2 \text{ where }\delta\text{ solves }\frac{1}{2}(1-\Phi(\frac{\sqrt{2}\delta}{\sigma_+}))-2\gamma\delta=0.
\end{align*}
As $\gamma$ describes $\R^*_+$, the associated $\delta$ describes also $\R_+^*$ so the problem can be rewritten
\begin{align*}
    \underset{\delta>0}{\sup}\, \frac{1}{4}(1-\Phi(\frac{\sqrt{2}\delta}{\sigma_+}))\delta .
\end{align*}
With the ansatz $\delta=\frac{\sigma y}{\sqrt{2}}$, $y>0$ we find from the first order optimality condition that the optimum $y$ is unique and that it solves $1-\Phi(y)-y\Phi'(y)=0$.
%
\end{proof}

\section{Numerical results}\label{sec::num}
\subsection{Suggestions for parameters calibration in practice}
In this model, the main parameters to be calibrated are $\sigma_-$, $\sigma_+$ and $\rho$. 

\begin{itemize}
\item The parameter $\sigma_+$ is the confidence level of the player suffering from the highest level of uncertainty on the efficient price. We can choose $\sigma_+$ equals to the volatility between two auctions, while $\sigma_-$ can be chosen small enough. 
\item The ratio $\frac{\sigma_-}{\sigma_+}$ can be seen as a measure of the symmetry of the ability of the players to estimate the efficient price. If $\sigma_-=\sigma_+$ then no player is particularly accurate compared to the others. In this case $\rho=1$ describes the situation where the players have the same forecast. The case $\rho\neq 1$ would then describe a market with some heterogeneity in the beliefs (trend-follower versus contrarian for example). If $\sigma_-<\sigma_+$ then some player might have better information. The ratio $\frac{\sigma_-}{\sigma_+}$ is similar to the ratio $r$ in \cite{saliba}. The second and fourth results of Proposition \ref{prop::stats} suggest that this ratio should be calibrated on the market using the observed spreads as well as $\rho$.
\end{itemize}

In the following we take $\sigma_- = 0.1$, $\sigma_+=1.1$ and $\rho=0$.

\subsection{Optimal strategies}
In Figure \ref{optimaldelta}, we plot  the optimal strategies $\delta$ as a function of $\sigma_i$ for different values of $\gamma$.
\begin{figure}[H]
    \centering
    \includegraphics[width=12cm]{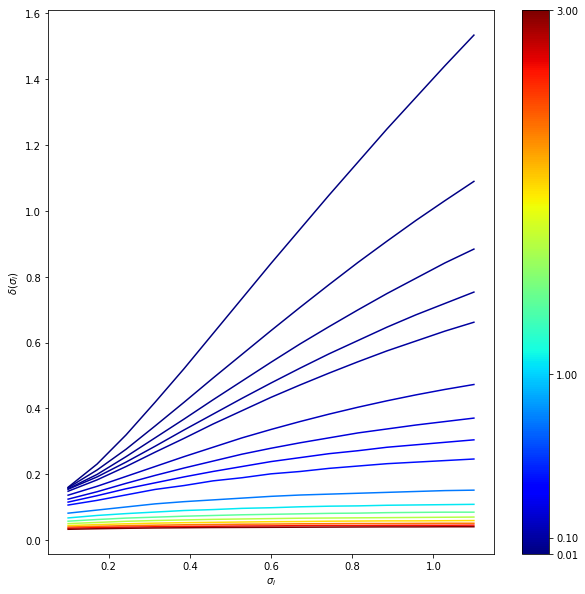}
    \caption{Optimal strategies $\delta(\sigma_i)$ in a NE. The colorbar gives possible values for the parameter $\gamma$.}\label{optimaldelta}
\end{figure}
We first observe that the optimal strategies are increasing in $\sigma_i$. This result is natural: the more uncertain the market, the greater the spread. Market makers increase their spread when the volatility is high. We also note that the optimal strategies of the market makers are not linear functions of the uncertainty. Finally, we see that the greater the fees, the smaller the spreads. This is a result of the incentive mechanism put in place by the exchange to reduce the difference between the price sent by the market makers and their own private estimates. 

\subsection{Optimal incentives}
We now turn to the optimal fees propose by the exchange to increase the liquidity. We compute in Figure \ref{figureoptfee} $\E[\gamma\delta(\sigma_i)^2]$ for multiple choices of $\gamma$ in NEs.
\begin{figure}[H]
    \centering
    \includegraphics{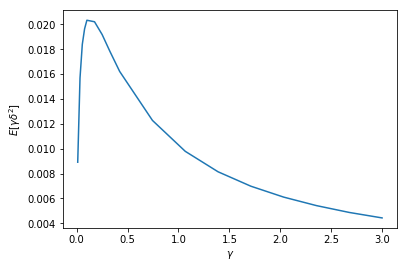}
    \caption{Penalties paid $\E[\gamma\delta(\sigma_i)^2]$.}\label{figureoptfee}
\end{figure}
With this set of parameters, the condition ensuring the existence of a Nash equilibrium is $\gamma\geq 1.5$. Therefore, based on Figure \ref{figureoptfee}, the exchange should pick $\gamma\approx 1.5$. It is  larger than the optimal $\gamma\simeq 0.1$ given by Theorem \ref{th::excheasy} if $\sigma_-=\sigma_+=1.1$.

\section{Extension to half-linear controls}\label{sec::ext}

\comment{\subsection{Adding noise traders}

Should be add noise traders? For example it could model the fact that small spreads might attract new players.\\

We consider an extension in which the players can, with some probability $p$, have a confidence $\sigma_{++}>\sigma_+$ in his prior. This player is penalized if he does not manage to trade. We have now
\begin{align*}
    e^{\sigma_a}_\delta(x)= &(1-p)\frac{1}{\sigma_+-\sigma_-}\int_{\sigma_-}^{\sigma^+}\big(\frac{x-\delta(\sigma_b)}{2}(1-\Phi(\frac{x+\delta(\sigma_b)}{\sqrt{\sigma_a^2+\sigma_b^2}}))+\sqrt{\sigma_a^2+\sigma_b^2}(\frac{1}{2}-\frac{\frac{1}{\sigma_b^2}}{\frac{1}{\sigma_a^2}+\frac{1}{\sigma_b^2}})\Phi'(\frac{x+\delta(\sigma_b)}{\sqrt{\sigma_a^2+\sigma_b^2}})\big)d\sigma_b\\
    &+p\big(\frac{x-\delta(\sigma_{++})}{2}(1-\Phi(\frac{x+\delta(\sigma_{++})}{\sqrt{\sigma_a^2+\sigma_{++}^2}}))+\sqrt{\sigma_a^2+\sigma_{++}^2}(\frac{1}{2}-\frac{\frac{1}{\sigma_{++}^2}}{\frac{1}{\sigma_a^2}+\frac{1}{\sigma_{++}^2}})\Phi'(\frac{x+\delta(\sigma_{++})}{\sqrt{\sigma_a^2+\sigma_{++}^2}})\big)\\
    &-\gamma\mathbf{1}_{\sigma_a=\sigma_{++}}((1-p)\frac{1}{\sigma_+-\sigma_-}\int_{\sigma_-}^{\sigma^+}(1-\Phi(\frac{x+\delta(\sigma_b)}{\sqrt{\sigma_a^2+\sigma_b^2}}))d\sigma_b+p(1-\Phi(\frac{x+\delta(\sigma_{++})}{\sqrt{\sigma_a^2+\sigma_{++}^2}})))
\end{align*}

\subsection{Half-linear controls}}

We consider a variant of our model to better account for the fact that each trader can use multiple limit orders. Player $a$ now has a demand function $p\mapsto K^a(p-P^a)\mathbf{1}_{p\geq P^a}$ and controls $P^a$ where $K^a$ is constant. Similarly Player $b$ has an offer function $p\mapsto K^b(P^b-p)\mathbf{1}_{p\leq P^b}$ and controls $P^b$, where $K^b$ is constant.\\

A trade once again happens if $P^a\leq P^b$, but in this case we now have the trade price which equals $\frac{K^aP^a+K^bP^b}{K^a+K^b}$ and the traded volume is $\frac{K^aK^b}{K^a+K^b}(P^b-P^a)$. So Player $a$ maximizes, for any $p_{\infty|a}$ and $\sigma_a$, his expected gain
\begin{align*}
    \E[\frac{K^aK^b}{K^a+K^b}(P^b-P^a)(\frac{K^aP^a+K^bP^b}{K^a+K^b}-P_\infty)\mathbf{1}_{P^a\leq P^b}|(P_{\infty|a}=p_{\infty|a},\boldsymbol\sigma_a=\sigma_a)].
\end{align*}

We have a first negative result proved in Appendix \ref{app::noNEk}.

\begin{proposition}\label{prop::noNEk}
If $K^a=K^b=\bar k>0$ there is no NE. 
\end{proposition}

As previously this result can be mitigated if 
both players pay a penalty $\gamma(P^i-P_{\infty|i})^2$. We have the following result proved in Appendix \ref{app::thmext}.
\begin{theorem}\label{th::quadspK}
Let $\gamma>0$. Assume that $\rho\leq\frac{\sigma_-}{\sigma_+}$ and that the slopes are fixed $K^a=K^b=\bar k$. Then
\begin{itemize}
    \item if $\sigma_-=\sigma_+$ there exist a unique NE with finite $\delta$. In that case $\delta(\sigma_+)$ is given by the unique positive root of the function $x\mapsto \frac{\bar k}{2}(-x(1-\Phi(\frac{\sqrt{2}x}{\sqrt{1-\rho}\sigma_+}))+\frac{\sigma_+\sqrt{1-\rho}}{\sqrt{2}}\Phi'(\frac{\sqrt{2}x}{\sqrt{1-\rho}\sigma_+}))-2\gamma x$ 
    \item if $\sigma_-<\sigma_+$, let $\gamma\geq \frac{\bar k}{2}\max(-(1-\Phi(z))+z\Phi'(z))$ there exists a NE with $\delta$ bounded by $C=\frac{\bar k}{4\gamma}\frac{1}{\sigma_+-\sigma_-}\int_{\sigma_-}^{\sigma_+}\bigg(\Sigma_\rho Q_\rho\Phi'(0)\bigg)d\sigma_b$.
\end{itemize}
\end{theorem}

\appendix

\section{Proof of Proposition \ref{prop::lim}}\label{app::prop::lim}
In the following, we set $$\epsilon_{ab} = \frac{P_{\infty|b}-P_{\infty|a}}{\boldsymbol\Sigma_\rho}$$ and $$\epsilon_\infty = (P_\infty-\frac{\frac{1}{\boldsymbol\sigma_a^2}P_{\infty|a}+\frac{1}{\boldsymbol\sigma_b^2}P_{\infty|b}-2\frac{\rho}{\boldsymbol\sigma_a\boldsymbol\sigma_b}\frac{P_{\infty|a}+P_{\infty|b}}{2}}{\frac{1}{\boldsymbol\sigma_a^2}+\frac{1}{\boldsymbol\sigma_b^2}-2\frac{\rho}{\boldsymbol\sigma_a\boldsymbol\sigma_b}})\sqrt{\frac{1}{\boldsymbol\sigma_a^2}+\frac{1}{\boldsymbol\sigma_b^2}-2\frac{\rho}{\boldsymbol\sigma_a\boldsymbol\sigma_b}}/\sqrt{1-\rho^2}.$$ We have the following lemma.
\begin{lemma}
     Conditioned by $\boldsymbol\sigma_a=\sigma_a$ and $\boldsymbol\sigma_b=\sigma_b$,  $\epsilon_{ab}$ and $\epsilon_{\infty}$ are independent standard Gaussian variables.
\end{lemma}
\begin{proof}
Given $\boldsymbol\sigma_a=\sigma_a$ and $\boldsymbol\sigma_b=\sigma_b$, we have
\begin{align*}
    \epsilon_{ab} =& \frac{\sigma_b\epsilon_b-\sigma_a\epsilon_a}{\Sigma_\rho},\\
    \epsilon_\infty =& -\frac{\frac{1}{\sigma_a}\epsilon_a+\frac{1}{\sigma_b}\epsilon_b-2\frac{\rho}{\sigma_a\sigma_b}\frac{\sigma_a\epsilon_a+\sigma_b\epsilon_b}{2}}{\sqrt{1-\rho^2}\sqrt{\frac{1}{\sigma_a^2}+\frac{1}{\sigma_b^2}-2\frac{\rho}{\sigma_a\sigma_b}}}
\end{align*}
and the conclusion follows.
\end{proof}

Conditional on $P_{\infty|a}$, $\epsilon_{ab}$ and $\epsilon_{\infty}$ are no longer independent, but they become independent when we take $v\to +\infty$. The following lemma comes directly from the conditional distribution of random Gaussian vectors.
\begin{lemma}\label{lem::condPa}
Conditioned by $P_{\infty|a}=P_{\infty|a}$, $\boldsymbol\sigma_a=\sigma_a$ and $\boldsymbol\sigma_b=\sigma_b$,  we have
\begin{align*}
    \begin{pmatrix}
    \Sigma_\rho\epsilon_{ab}\\
    -\sqrt{\frac{1}{\sigma_a^2}+\frac{1}{\sigma_b^2}-2\frac{\rho}{\sigma_a\sigma_b}}\epsilon_{\infty}
    \end{pmatrix}\sim\Nc\Bigg(\begin{pmatrix}
    P_{\infty|a}\frac{-\sigma_a^2+\rho\sigma_a\sigma_b}{v^2+\sigma_a^2}\\
    P_{\infty|a}\frac{\sqrt{1-\rho^2}}{v^2+\sigma_a^2}
    \end{pmatrix},\begin{pmatrix}
\Sigma_\rho^2-\frac{(-\sigma_a^2+\rho\sigma_a\sigma_b)^2}{v^2+\sigma_a^2} & \frac{(\sigma_a^2-\rho\sigma_a\sigma_b)\sqrt{1-\rho^2}}{v^2+\sigma_a^2} \\
\frac{(\sigma_a^2-\rho\sigma_a\sigma_b)\sqrt{1-\rho^2}}{v^2+\sigma_a^2} & \frac{1}{\sigma_a^2}+\frac{1}{\sigma_b^2}-2\frac{\rho}{\sigma_a\sigma_b}-\frac{1-\rho^2}{v^2+\sigma_a^2}
\end{pmatrix}\Bigg)\\
\end{align*}
and $(\epsilon_{ab}, \epsilon_{\infty})$ converges when $v\to \infty$ towards a centered Gaussian vector with unit variance.  
   
\end{lemma}
Notice that the limit distribution does not depend on $P_{\infty|a}$.\\

\begin{proof}
Given $P_{\infty|a}=p_{\infty|a}$, $\boldsymbol\sigma_a=\sigma_a$, we can write 
\begin{align*}
    \frac{P_{\infty|a}+P_{\infty|b}}{2}-P_\infty&=(p_{\infty|a}+P_{\infty|b}) (\frac{1}{2}-\boldsymbol Q^b_\rho)+\epsilon_\infty\frac{\sqrt{1-\rho^2}}{\sqrt{\frac{1}{\sigma_a^2}+\frac{1}{\boldsymbol\sigma_b^2}-2\frac{\rho}{\sigma_a\boldsymbol\sigma_b}}}\\
    &= \epsilon_\infty\frac{\sqrt{1-\rho^2}}{\sqrt{\frac{1}{\sigma_a^2}+\frac{1}{\boldsymbol\sigma_b^2}-2\frac{\rho}{\sigma_a\boldsymbol\sigma_b}}}+\epsilon_{ab}\boldsymbol\Sigma^b_\rho(\frac{1}{2}-\boldsymbol Q_\rho^b)
\end{align*}
and
\begin{align*}
    P_{\infty|a}-P_{\infty|b} = -\epsilon_{ab}\boldsymbol\Sigma^b_\rho.
\end{align*}
Therefore
\begin{align*}
    &\E[(\frac{P_{\infty|a}+x+P_{\infty|b}-\delta(\boldsymbol\sigma_b)}{2}-P_\infty)\mathbf{1}_{P_{\infty|a}+x\leq P_{\infty|b}-\delta(\sigma_b)}| (P_{\infty|a}=p_{\infty|a},\boldsymbol\sigma_a=\sigma_a)]\\
    &=\E[\big(\frac{x-\delta(\boldsymbol\sigma_b)}{2}+\epsilon_\infty\frac{\sqrt{1-\rho^2}}{\sqrt{\frac{1}{\sigma_a^2}+\frac{1}{\boldsymbol\sigma_b^2}-2\frac{\rho}{\sigma_a\boldsymbol\sigma_b}}}+\epsilon_{ab}\boldsymbol\Sigma^b_\rho(\frac{1}{2}-\boldsymbol Q_\rho^b)\big)\mathbf{1}_{\epsilon_{ab}\geq \frac{x+\delta(\boldsymbol\sigma_b)}{\boldsymbol\Sigma^b_\rho}}| (P_{\infty|a}=p_{\infty|a},\boldsymbol\sigma_a=\sigma_a)]\\
    &=\E[\big(\frac{x-\delta(\boldsymbol\sigma_b)}{2}+P_{\infty|a}\frac{\frac{-(1-\rho^2)}{v^2+\sigma_a^2}}{\frac{1}{\sigma_a^2}+\frac{1}{\boldsymbol\sigma_b^2}-2\frac{\rho}{\sigma_a\boldsymbol\sigma_b}}+P_{\infty|a}\frac{\frac{\sqrt{1-\rho^2}(-\sigma_a^2+\rho\sigma_a\boldsymbol\sigma_b)}{v^2+\sigma_a^2}}{\frac{1}{\sigma_a^2}+\frac{1}{\boldsymbol\sigma_b^2}-2\frac{\rho}{\sigma_a\boldsymbol\sigma_b}}\frac{\sqrt{\frac{1}{\sigma_a^2}+\frac{1}{\boldsymbol\sigma_b^2}-2\frac{\rho}{\sigma_a\boldsymbol\sigma_b}-\frac{1-\rho^2}{v^2+\sigma_a^2}}}{\sqrt{|\boldsymbol\Sigma^b_\rho|^2-\frac{(-\sigma_a^2+\rho\sigma_a\boldsymbol\sigma_b)^2}{v^2+\sigma_a^2}}}\\
    &\hspace{1em}+\epsilon_{ab}(-\chi(\sigma_a,\boldsymbol\sigma_b,v)+\boldsymbol\Sigma^b_\rho(\frac{1}{2}-\boldsymbol Q^b_\rho))\big)\mathbf{1}_{\epsilon_{ab}\geq \frac{x+\delta(\boldsymbol\sigma_b)}{\boldsymbol\Sigma^b_\rho}}| (P_{\infty|a}=p_{\infty|a},\boldsymbol\sigma_a=\sigma_a)]\\
    &= \frac{1}{\sigma_+-\sigma_-}\int_{\sigma_-}^{\sigma^+}\Bigg[\Big(\frac{x-\delta(\sigma_b)}{2}+p_{\infty|a} S(\sigma_a,\sigma_b,\rho)+p_{\infty|a}\frac{\frac{-\sigma_a^2+\rho\sigma_a\sigma_b}{v^2+\sigma_a^2}}{\Sigma_\rho}\big(-\chi(\sigma_a,\sigma_b,v)+\Sigma_\rho(\frac{1}{2}-Q_\rho)\big)\Big)\\
    &\hspace{9em}\times \Big(1-\Phi\big(\frac{x+\delta(\sigma_b)-p_{\infty|a}\frac{-\sigma_a^2+\rho\sigma_a\sigma_b}{v^2+\sigma_a^2}}{\sqrt{\Sigma_\rho^2-\frac{(-\sigma_a^2+\rho\sigma_a\sigma_b)^2}{v^2+\sigma_a^2}}}\big)\Big)\\
    &\vspace{1cm}\hspace{5em}+\frac{\sqrt{\Sigma_\rho^2-\frac{(-\sigma_a^2+\rho\sigma_a\sigma_b)^2}{v^2+\sigma_a^2}}}{\Sigma_\rho}(-\chi(\sigma_a,\sigma_b,v)+\Sigma_\rho(\frac{1}{2}-Q_\rho))\Phi'(\frac{x+\delta(\sigma_b)-p_{\infty|a}\frac{-\sigma_a^2+\rho\sigma_a\sigma_b}{v^2+\sigma_a^2}}{\sqrt{\Sigma_\rho^2-\frac{(-\sigma_a^2+\rho\sigma_a\sigma_b)^2}{v^2+\sigma_a^2}}})\Bigg]d\sigma_b
\end{align*}
by setting
\[S(\sigma_a,\sigma_b,\rho):=\frac{\frac{-(1-\rho^2)}{v^2+\sigma_a^2}}{\frac{1}{\sigma_a^2}+\frac{1}{\sigma_b^2}-2\frac{\rho}{\sigma_a\sigma_b}}+ \frac{\frac{\sqrt{1-\rho^2}(-\sigma_a^2+\rho\sigma_a\sigma_b)}{v^2+\sigma_a^2}}{\frac{1}{\sigma_a^2}+\frac{1}{\sigma_b^2}-2\frac{\rho}{\sigma_a\sigma_b}}\frac{\sqrt{\frac{1}{\sigma_a^2}+\frac{1}{\sigma_b^2}-2\frac{\rho}{\sigma_a\sigma_b}-\frac{1-\rho^2}{v^2+\sigma_a^2}}}{\sqrt{\Sigma_\rho^2-\frac{(-\sigma_a^2+\rho\sigma_a\sigma_b)^2}{v^2+\sigma_a^2}}},\]

$$\chi(\sigma_a,\sigma_b,v)=\frac{\Sigma_\rho\frac{(\sigma_a^2-\rho\sigma_a\sigma_b)\sqrt{1-\rho^2}}{v^2+\sigma_a^2}}{(\Sigma_\rho^2-\frac{(-\sigma_a^2+\rho\sigma_a\sigma_b)^2}{v^2+\sigma_a^2})\sqrt{\Sigma_\rho^2\frac{(-\sigma_a^2+\rho\sigma_a\sigma_b)^2}{v^2+\sigma_a^2}}(\frac{1}{\sigma_a^2}+\frac{1}{\sigma_b^2-2\frac{\rho}{\sigma_a\sigma_b}})},$$

and by using Lemma \ref{lem::condPa}. Now, as $\chi(\sigma_a,\sigma_b,v)\underset{v\to+\infty}{\longrightarrow}0$, we have 
\begin{align*}
    &\E[(\frac{P_{\infty|a}+x+P_{\infty|b}-\delta(\boldsymbol\sigma_b)}{2}-P_\infty)\mathbf{1}_{P_{\infty|a}+x\leq P_{\infty|b}-\delta(\boldsymbol\sigma_b)}| (P_{\infty|a}=p_{\infty|a},\boldsymbol\sigma_a=\sigma_a)]\\
    \underset{v\to+\infty}{\longrightarrow}&\frac{1}{\sigma_+-\sigma_-}\int_{\sigma_-}^{\sigma^+}\big(\frac{x-\delta(\sigma_b)}{2}(1-\Phi(\frac{x+\delta(\sigma_b)}{\Sigma_\rho}))+\Sigma_\rho(\frac{1}{2}-Q_\rho)\Phi'(\frac{x+\delta(\sigma_b)}{\Sigma_\rho})\big)d\sigma_b.
\end{align*}
Notice that the convergence is uniform in $x$ as $\Phi$ and $\Phi'$ are bounded.
\end{proof}

\section{Proof of Theorem \ref{th::quadsp}}\label{app::th::quadsp}
Using the same notations as previously we have
\begin{align*}
    \E[(P^a-P_{\infty|a})^2| (P_{\infty|a},\sigma_a)] = x^2
\end{align*}
so
\begin{align*}
    e^{\sigma_a}_\delta(x)=&\frac{1}{\sigma_+-\sigma_-} \int_{\sigma_-}^{\sigma^+}\big(\frac{x-\delta(\sigma_b)}{2}(1-\Phi(\frac{x+\delta(\sigma_b)}{\Sigma_\rho}))+\Sigma_\rho(\frac{1}{2}-Q_\rho)\Phi'(\frac{x+\delta(\sigma_b)}{\Sigma_\rho})\big)d\sigma_b- \gamma x^2
\end{align*}
and 
\begin{align*}
    (e^{\sigma_a}_\delta)'(x)=&  \frac{1}{\sigma_+-\sigma_-}\int_{\sigma_-}^{\sigma^+}\big(\frac{1}{2}(1-\Phi(\frac{x+\delta(\sigma_b)}{\Sigma_\rho}))\\
    &\hspace{3cm}+\frac{1}{\Sigma_\rho}(-x\tilde Q_\rho+\delta(\sigma_b)Q_\rho)\Phi'(\frac{x+\delta(\sigma_b)}{\Sigma_\rho})\big)d\sigma_b-2\gamma x.
\end{align*}
If $\sigma_-=\sigma_+$, $\delta$ describes a NE equilibrium if and only if
\begin{align*}
    \frac{1}{2}(1-\Phi(\frac{\sqrt{2}\delta(\sigma_+)}{\sigma_+\sqrt{1-\rho}}))-2\gamma \delta(\sigma_+)=0,
\end{align*}
which can always be achieved for some unique $\delta(\sigma_+)>0$ if $\gamma>0$. \\

If $\sigma_-<\sigma_+$, we start by observing that it makes sense to consider $\delta$ non-negative. Indeed if $0\leq \delta\leq C$ we have $e^{\sigma_a}_\delta(x)\underset{x\to+\infty}{\longrightarrow}-\infty$ and $(e^{\sigma_a}_\delta)'(y)> 0$ for all $y\leq 0$. By the asymptotic expansion $1-\Phi(x)\underset{x\to +\infty}{\sim} \frac{\Phi'(x)}{x}$ we also get $(e^{\sigma_a}_\delta)'(x)\underset{x\to +\infty}{\longrightarrow}-\infty$ so by elementary computations $e^{\sigma_a}_\delta$ attains its maximum on $\R_+^*$.\\

By computing the second derivative we get
\begin{align*}
    (e^{\sigma_a}_\delta)''(x)=& \frac{1}{\sigma_+-\sigma_-}\int_{\sigma_-}^{\sigma^+}\big(-(\frac{1}{2}+\tilde Q_\rho+\frac{1}{\Sigma_\rho^2}Q_\rho\delta(\sigma_b)^2)-\frac{\delta(\sigma_b)}{\Sigma_\rho^2}Q_\rho x+\frac{1}{\Sigma_\rho^2}\tilde Q_\rho x^2\big)\frac{\Phi'(\frac{x+\delta(\sigma_b)}{\Sigma_\rho})}{\Sigma_\rho}d\sigma_b-2\gamma.
\end{align*}
Each quadratic function in the integral is non-positive on an interval
\begin{align*}
    &[\frac{1}{2}\delta(\sigma_b)\frac{(\frac{\sigma_a}{\sigma_b})^2-1}{1-\rho\frac{\sigma_a}{\sigma_b}}\\
    &\pm\sqrt{(\frac{1}{2}\delta(\sigma_b)\frac{(\frac{\sigma_a}{\sigma_b})^2-1}{1-\rho\frac{\sigma_a}{\sigma_b}})^2+(\frac{1}{2}+\tilde Q_\rho+\frac{1}{\sigma_a^2+\sigma_b^2-2\rho\sigma_a\sigma_b} Q_\rho\delta(\sigma_b)^2)\frac{\Sigma_\rho^2}{\sigma_b^2-\rho\sigma_a\sigma_b}}].
\end{align*}
Take $C=2\lambda\frac{1-\rho\frac{\sigma_+}{\sigma_-}}{1-(\frac{\sigma_-}{\sigma_+})^2}\sqrt{\frac{1}{2}\frac{(1-\rho^2)^2\sigma_-^4}{\sigma_+^2}}$ for some $\lambda\in(0,1)$. Then the intersection of those intervals contains the interval $[0,(1-\lambda)\sqrt{\frac{1}{2}\frac{(1-\rho^2)^2\sigma_-^4}{\sigma_+^2}}]$. In particular if $\lambda=\frac{1}{2\frac{1-\rho\frac{\sigma_+}{\sigma_-}}{1-(\frac{\sigma_-}{\sigma_+})^2}+1}$ we have $C = (1-\lambda)\sqrt{\frac{1}{2}\frac{(1-\rho^2)^2\sigma_-^4}{\sigma_+^2}}$ and so $(e^{\sigma_a}_\delta)''<0$ on $(0,C)$ for all $\sigma_a$.\\

Let $c\geq C$. We have for all $\sigma_a$
\begin{align*}
    (e^{\sigma_a}_\delta)'(c)\leq & \underset{u\in[\sigma_-,\sigma_+]}{\max}\frac{1}{\sigma_+-\sigma_-}\int_{\sigma_-}^{\sigma^+}\big(\frac{1}{2}(1-\Phi(\frac{c}{\sqrt{u^2+\sigma_b^2-2\rho u\sigma_b}}))\\
    &+\frac{1}{\sqrt{u^2+\sigma_b^2-2\rho u\sigma_b}}\underset{F\in{[0,C]}}{\max}(-cG(u;\sigma_b)+F\tilde G(u;\sigma_b))\Phi'(\frac{c+F}{\sqrt{u^2+\sigma_b^2-2\rho u\sigma_b}})\big)d\sigma_b\\
    &-2\gamma c\\
    &\leq \frac{1}{\sigma_+-\sigma_-}\int_{\sigma_-}^{\sigma^+}\frac{1}{2}(1-\Phi(\frac{ C}{\sqrt{\sigma_+^2+\sigma_b^2-2\rho u\sigma_b}}))d\sigma_b\\
    &+ C\frac{1}{\sigma_+-\sigma_-}\int_{\sigma_-}^{\sigma^+}\frac{1}{\sqrt{u^2+\sigma_b^2}}\tilde G(u;\sigma_b)\Phi'(\frac{C}{\sqrt{u^2+\sigma_b^2-2\rho u\sigma_b}})d\sigma_b-2\gamma C\\
    &\leq \frac{1}{2}(1-\Phi(\frac{ C}{\sqrt{2}\sigma_+}))+ C(\frac{1}{\sqrt{2}\sigma_-}\Phi'(\frac{ C}{\sqrt{2}\sigma_+})-2\gamma )
\end{align*}
by setting
\[ G(u;\sigma_b)=
\frac{\frac{1}{u^2}-\rho\frac{1}{u}\frac{1}{\sigma_b}}{\frac{1}{u^2}+\frac{1}{\sigma_b^2}-2\rho\frac{1}{u}\frac{1}{\sigma_b}}\]
and
\[ \tilde G(u;\sigma_b)=\frac{\frac{1}{\sigma_b^2}-\rho\frac{1}{u}\frac{1}{\sigma_b}}{\frac{1}{u^2}+\frac{1}{\sigma_b^2}-2\rho\frac{1}{u}\frac{1}{\sigma_b}}.
\] So for $\gamma$ large enough we have
\begin{align*}
    \gamma\geq \frac{1}{4C}(1-\Phi(\frac{ C}{\sqrt{2}\sigma_+}))+\frac{1}{2\sqrt{2}\sigma_-}\Phi'(\frac{ C}{\sqrt{2}\sigma_+})
\end{align*}
so $(e^{\sigma_a}_\delta)'(c)\leq 0$ for all $\sigma_a$, $c\geq  C$.\\

In this case, if we let $\Cc_{[0,c]}$ be the set of continuous functions on $[\sigma_-,\sigma_+]$ with values on $[0,C]$, then $\delta\in \Cc_{[0,C]}\mapsto (\sigma_a\mapsto \arg\max e^{\sigma_a}_\delta)\in \Cc_{[0,C]}$ is continuous in the $L^\infty$ norm. We consider a sequence $\delta_n\in\Cc_{[0,c]}$ converge to  $\delta\in\Cc_{[0,c]}$ in the $L^\infty$ norm. We first show that 
\begin{align*}
    \underset{(\sigma_a,x)\in[\sigma_-,\sigma_+]\times[0,C]}{\sup}|e_\delta^{\sigma_a}(x)-e_{\delta_n}^{\sigma_a}(x)|\underset{n\to +\infty}{\longrightarrow}0.
\end{align*}
To see that, we can observe that, if $\|\delta-\delta_n\|_{L^\infty}<\nu$, then, using the intermediate values theorem, we have, for any $(\sigma_a,x)\in[\sigma_-,\sigma_+]\times[0,C]$, 
\begin{align*}
    |e_\delta^{\sigma_a}(x)-e_{\delta_n}^{\sigma_a}(x)|\leq& \frac{1}{\sigma_+-\sigma_-}\int_{\sigma_-}^{\sigma^+}\big(\frac{\nu}{2}+\underset{g\in[0,C]}{\sup}|\frac{x-g}{2}\frac{\nu}{\Sigma_\rho}\Phi'(\frac{x+g}{\Sigma_\rho})|\\
    &+|(\frac{1}{2}- Q_\rho)|\nu\underset{g\in[0,C]}{\sup}|\Phi''(\frac{x+\delta(\sigma_b)}{\Sigma_\rho})|\big)d\sigma_b\\
    \leq& \frac{1}{\sigma_+-\sigma_-}\int_{\sigma_-}^{\sigma^+}\big(\frac{\nu}{2}+C\nu\sup|\Phi'|\frac{1}{\Sigma_\rho}+|(\frac{1}{2}- Q_\rho)|\nu\sup|\Phi''|\big)d\sigma_b
\end{align*}
so 
\begin{align*}
    \underset{(\sigma_a,x)\in[\sigma_-,\sigma_+]\times[0,C]}{\sup}|e_\delta^{\sigma_a}(x)-e_{\delta_n}^{\sigma_a}(x)|\leq K\nu
\end{align*}
where $K$ is some constant which depends only on $C, \sigma_-,\sigma_+$. This proves that
\begin{align*}
    \underset{(\sigma_a,x)\in[\sigma_-,\sigma_+]\times[0,C]}{\sup}|e_\delta^{\sigma_a}(x)-e_{\delta_n}^{\sigma_a}(x)|\underset{n\to +\infty}{\longrightarrow}0.
\end{align*}

In particular, $e_{\delta_n}^{\sigma_a}$ converges uniformly to $e_\delta^{\sigma_a}$ for all $\sigma_a\in[\sigma_-,\sigma_+]$. By strong concavity on $[0,C]$ and classical results on epi-convergence, we have that $\arg\max e_{\delta_n}^{\sigma_a}$ converges to $\arg\max e_\delta^{\sigma_a}$. The continuity of $\sigma_a\mapsto\arg\max e_\delta^{\sigma_a}$ also comes from the same epi-convergence argument.\\

We now prove that the convergence is uniform. We suppose that it is not. Let $\eta>0$ and $\sigma_n\in[\sigma_-,\sigma_+]$ such that $|\arg\max e_\delta^{\sigma_n}-\arg\max e_{\delta_n}^{\sigma_n}|\geq\eta$ for all $n\in\N$. We can assume that $\sigma_n\longrightarrow\sigma_a$. The function $e_{\delta_n}^{\sigma_n}$ is twice derivable at $\arg\max e_{\delta_n}^{\sigma_n}$ and $e_{\delta_n}^{\sigma_n}\leq -2\gamma$ so we have
\begin{align*}
    e_{\delta_n}^{\sigma_n}(\arg\max e_{\delta_n}^{\sigma_n})-e_{\delta_n}^{\sigma_n}(\arg\max e_{\delta}^{\sigma_n})\geq \gamma(\arg\max e_{\delta_n}^{\sigma_n}-\arg\max e_{\delta}^{\sigma_n})^2\geq \gamma\eta^2.
\end{align*}
However \[e_{\delta_n}^{\sigma_n}(\arg\max e_{\delta_n}^{\sigma_n})=\max e_{\delta_n}^{\sigma_n}\longrightarrow\max e_{\delta}^{\sigma_a}\] and \[e_{\delta_n}^{\sigma_n}(\arg\max e_{\delta}^{\sigma_n})=e_{\delta}^{\sigma_n}(\arg\max e_{\delta}^{\sigma_n})+(e_{\delta_n}^{\sigma_n}(\arg\max e_{\delta}^{\sigma_n})-e_{\delta}^{\sigma_n}(\arg\max e_{\delta}^{\sigma_n})),\] with \[e_{\delta}^{\sigma_n}(\arg\max e_{\delta}^{\sigma_n})=\max e_{\delta}^{\sigma_n}\longrightarrow \max e_{\delta}^{\sigma_a}\] and \[|e_{\delta_n}^{\sigma_n}(\arg\max e_{\delta}^{\sigma_n})-e_{\delta}^{\sigma_n}(\arg\max e_{\delta}^{\sigma_n})|\longrightarrow 0.\] So $e_{\delta_n}^{\sigma_n}(\arg\max e_{\delta_n}^{\sigma_n})-e_{\delta_n}^{\sigma_n}(\arg\max e_{\delta}^{\sigma_n})\longrightarrow 0$ and we obtain a contradiction.\\

Finally we can apply the Schauder fixed-point theorem to show that there exists a NE.

\section{Proof of Section \ref{sec::ext}}

\subsection{Proof of Proposition \ref{prop::noNEk}}\label{app::noNEk}

We have
\begin{align*}
    e^\delta_{\sigma_a}(x) =& \frac{1}{\sigma_+-\sigma_-}\int_{\sigma_-}^{\sigma_+}\frac{\bar k}{2}\bigg(\big(\Sigma_\rho^2(\frac{1}{2}-Q_\rho)-\frac{x -\delta(\sigma_b)}{2}(\delta(\sigma_b)+x)\big)(1-\Phi(\frac{x+\delta(\sigma_b)}{\Sigma_\rho}))\\
    &\hspace{7em}+\Sigma_\rho\frac{x - \delta(\sigma_b)}{2}\Phi'(\frac{x+\delta(\sigma_b)}{\Sigma_\rho})\bigg)d\sigma_b
\end{align*}
so

\begin{align*}
    (e^{\delta}_{\sigma_a}) '(x) =& \frac{1}{\sigma_+-\sigma_-}\int_{\sigma_-}^{\sigma_+}\frac{\bar k}{2}\bigg(-x(1-\Phi(\frac{x+\delta(\sigma_b)}{\Sigma_\rho}))+\Sigma_\rho Q_\rho\Phi'(\frac{x+\delta(\sigma_b)}{\Sigma_\rho})\bigg)d\sigma_b
\end{align*}

and we have once again that $(e^{\delta}_{\sigma_a}) '(x)>0$ if $x\leq 0$. If $\delta$ gives a NE we should have 
\begin{align*}
    \frac{1}{\sigma_+-\sigma_-}\int_{\sigma_-}^{\sigma_+}(e^{\delta}_{\sigma_a} )'(\delta(\sigma_a))d\sigma_a =0
\end{align*}
so
\begin{align*}
    &\int_{\sigma_-}^{\sigma_+}\int_{\sigma_-}^{\sigma_+}(-\delta(\sigma_a)(1-\Phi(\frac{\delta(\sigma_a)+\delta(\sigma_b)}{\Sigma_\rho}))+\Sigma_\rho Q_\rho\Phi'(\frac{\delta(\sigma_a)+\delta(\sigma_b)}{\Sigma_\rho}))d\sigma_a d\sigma_b=0.
\end{align*}
Summing this and the symmetric equation we get
\begin{align*}
    &\int_{\sigma_-}^{\sigma_+}\int_{\sigma_-}^{\sigma_+}(-(\delta(\sigma_a)+\delta(\sigma_b))(1-\Phi(\frac{\delta(\sigma_a)+\delta(\sigma_b)}{\Sigma_\rho}))+\Sigma_\rho\Phi'(\frac{\delta(\sigma_a)+\delta(\sigma_b)}{\Sigma_\rho}))d\sigma_a d\sigma_b=0
\end{align*}
which is impossible as $\frac{x(1-\Phi(x))}{\Phi'(x)}<1$ for $x>0$.\\

\subsection{Proof of Theorem \ref{th::quadspK}}\label{app::thmext}

We prove the second point, as the first one is straightforward from the proof of Proposition \ref{prop::noNEk}, by analogy with Theorem \ref{th::quadsp}.\\

We have
\begin{align*}
    (e^{\delta}_{\sigma_a}) '(x) =& \frac{1}{\sigma_+-\sigma_-}\int_{\sigma_-}^{\sigma_+}\frac{\bar k}{2}\bigg(-x(1-\Phi(\frac{x+\delta(\sigma_b)}{\Sigma_\rho}))+\Sigma_\rho Q_\rho\Phi'(\frac{x+\delta(\sigma_b)}{\Sigma_\rho})\bigg)d\sigma_b-2\gamma x
\end{align*}
so 
\begin{align*}
    (e^{\delta}_{\sigma_a}) ''(x) =& \frac{1}{\sigma_+-\sigma_-}\int_{\sigma_-}^{\sigma_+}\frac{\bar k}{2}\bigg(-(1-\Phi(\frac{x+\delta(\sigma_b)}{\Sigma_\rho}))+\frac{1}{\Sigma_\rho}(\tilde Q_\rho x- Q_\rho\delta(\sigma_b))\Phi'(\frac{x+\delta(\sigma_b)}{\Sigma_\rho})\bigg)d\sigma_b-2\gamma.
\end{align*}
We need to find a constant $C$ and a parameter $\gamma$ such that $(e^{\delta}_{\sigma_a}) '(x)<0$ and $(e^{\delta}_{\sigma_a}) ''(y)<0$ if $x\geq C\geq y\geq 0$ and $0\leq\delta\leq C$. We have
\begin{align*}
    (e^{\delta}_{\sigma_a}) ''(y) \leq& \frac{1}{\sigma_+-\sigma_-}\int_{\sigma_-}^{\sigma_+}\frac{\bar k}{2}\bigg(-(1-\Phi(\frac{y+\delta(\sigma_b)}{\Sigma_\rho}))+\frac{y+\delta(\sigma_b)}{\Sigma_\rho}\Phi'(\frac{y+\delta(\sigma_b)}{\Sigma_\rho})\bigg)d\sigma_b-2\gamma\\
    \leq & \frac{\bar k}{2}\max(-(1-\Phi(z))+z\Phi'(z))-2\gamma\leq 0
\end{align*}
and
\begin{align*}
    (e^{\delta}_{\sigma_a}) '(x) \leq& \frac{1}{\sigma_+-\sigma_-}\int_{\sigma_-}^{\sigma_+}\frac{\bar k}{2}\bigg(\Sigma_\rho Q_\rho\Phi'(\frac{C}{\Sigma_\rho})\bigg)d\sigma_b-2\gamma C
\end{align*}
and the conclusion follows easily from the assumptions and the arguments of the proof of Theorem \ref{th::quadsp}.

\end{document}